\documentclass[letterpaper, 10 pt, conference]{ieeeconf}  
\usepackage[utf8]{inputenc}
\usepackage[T1]{fontenc}

\IEEEoverridecommandlockouts                              

\usepackage{graphics} 
\usepackage{epsfig} 
\usepackage{mathptmx} 
\usepackage{times} 
\usepackage{amsmath} 
\usepackage{amssymb}  
\usepackage{algorithmic}
\usepackage{graphicx}
\usepackage{textcomp}
\usepackage{algorithm}
\usepackage{hyperref}
\usepackage{xcolor}

\title{\LARGE \bf
SCORE: Statistical Certification of Regions of Attraction via Extreme Value Theory
}

\author{Pietro Zanotta$^{1}$ and Panos Stinis$^{2}$ and Ján Drgoňa$^{1}$
 \thanks{This work was supported by the U.S. Department of Energy (DOE), Office of Science, Advanced Scientific Computing Research (ASCR) program, under the Uncertainty Quantification for Multifidelity Operator Learning (MOLUcQ) project (Project No. 81739). }
\thanks{$^{1}$Johns Hopkins University. (pzanott1@jh.edu, jdrgona1@jh.edu)}%
\thanks{$^{2}$Pacific Northwest National Laboratory. (panagiotis.stinis@pnnl.gov)}%
}

\usepackage{amsthm}

\newtheoremstyle{break}
  {}
  {}
  {\itshape}
  {}
  {\bfseries}
  {.}
  {\newline}
  {}

\theoremstyle{break}

\newtheorem{assumption}{Assumption}
\newtheorem{remark}{Remark}
\newtheorem{lemma}{Lemma}
\newtheorem{theorem}{Theorem}
\newtheorem{corollary}{Corollary}

\begin{document}

\maketitle
\thispagestyle{empty}
\pagestyle{empty}


\begin{abstract}

Certifying the Region of Attraction (ROA) for high-dimensional nonlinear dynamical systems remains a severe computational bottleneck. Traditional deterministic verification methods provide hard guarantees but suffer from the curse of dimensionality, typically failing to scale beyond 20 dimensions. To overcome these limitations, we propose SCORE, a statistical certification framework that shifts from seeking deterministic guarantees to bounding the worst-case safety violation with high statistical confidence. By integrating Projected Stochastic Gradient Langevin Dynamics (PSGLD) with Extreme Value Theory (EVT), we frame ROA certification as a constrained extreme-value estimation problem. Under stationary sampling and a regular local-geometry condition around the global maximum, we show that the Lyapunov derivative belongs to the Weibull maximum domain of attraction. Its finite right endpoint enables statistical estimation of the global maximum of the Lyapunov derivative and construction of an upper confidence bound, conditional on the sampling and inference assumptions. Numerical experiments validate that our EVT-based approach achieves certification tightness competitive to exact Sum of Squares programming on a 2D Van der Pol benchmark. Furthermore, we demonstrate strong scalability by successfully applying the statistical verification procedure to a dense, unstructured 500-dimensional ODE system at a nominal confidence level of 99.99\%, effectively bypassing the severe combinatorial constraints that limit existing formal verification pipelines.
\end{abstract}

\begin{figure}[h!]
    \centering
    \includegraphics[width=0.9\linewidth]{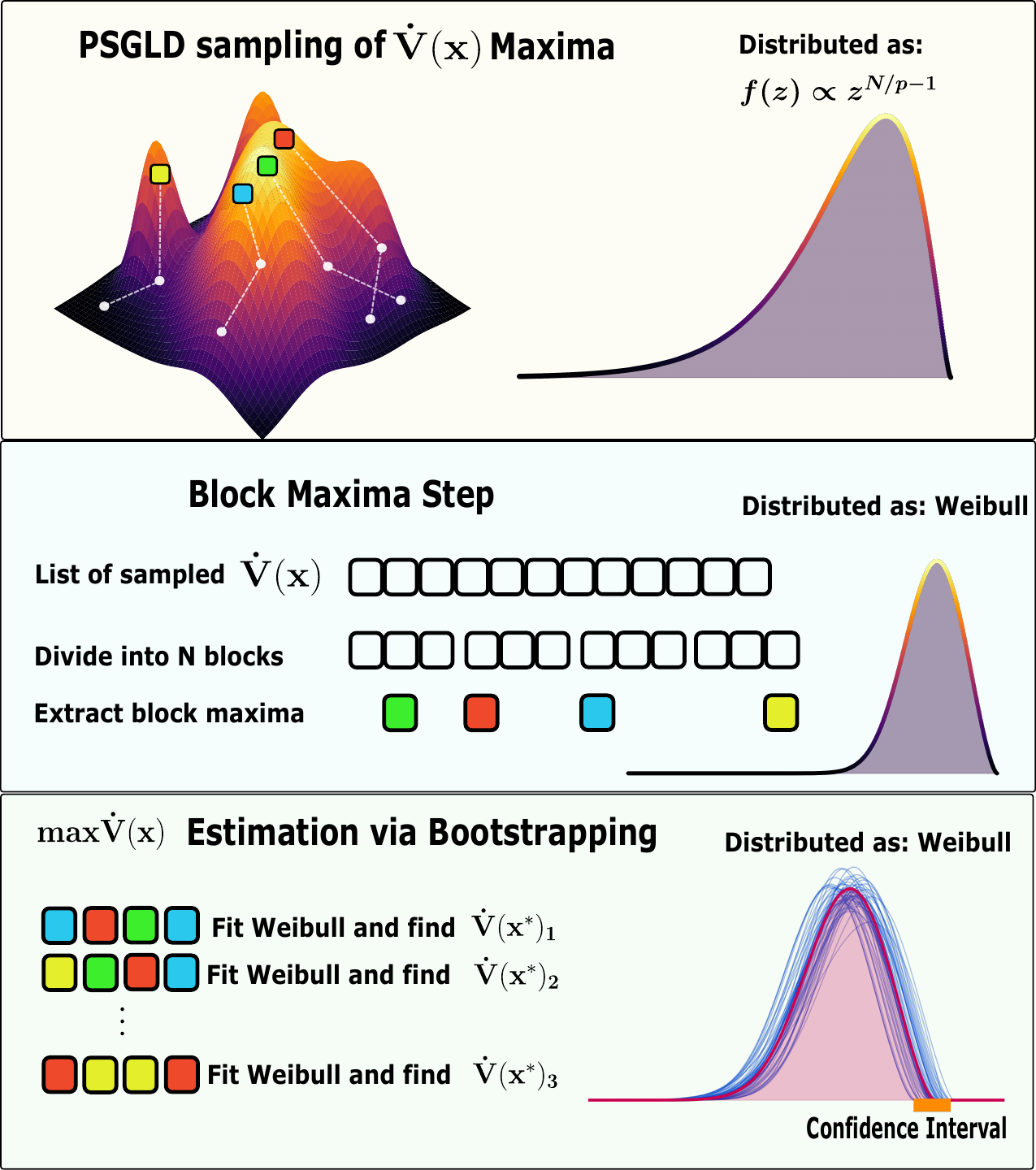}    
    \caption{Overview of the SCORE pipeline. PSGLD explores large values of $\dot{V}$ on the compact verification domain. Under the stated stationary-sampling and local-geometry assumptions, block maxima asymptotically follow a Weibull-class GEV law. Bootstrap resampling estimates an upper confidence bound on its finite endpoint; a negative bound yields a conditional statistical certificate.}
\end{figure}

\section{INTRODUCTION}
\label{sec:introduction}
The deployment of safety-critical control systems requires rigorous guarantees of stability across a wide range of initial conditions. For continuous-time nonlinear dynamical systems, estimating the Region of Attraction (ROA) is a fundamental problem that is addressed via Lyapunov functions \cite{khalil2002nonlinear}. Lyapunov's method yields a formal certificate of asymptotic stability by identifying a positive definite function $V$. 
However, the core requirement of this method, the rigorous certification of the decrease condition ($\dot{V}(x) < 0$) across a continuous state space, remains a severe computational challenge in high-dimensions. Traditional deterministic verification techniques, such as Sum-of-Squares (SOS) programming \cite{antonis_papachristodoulou_sprajna_2005, A_theorem_of_the_alternative_for_SOS_Lyapunov_functions} and Satisfiability Modulo Theories (SMT) \cite{chang_roohi_gao_2019, gao_kong_clarke_2013}, provide deterministic stability guarantees but scale poorly with state dimension. SOS programming suffers from combinatorial matrix explosions for high-degree polynomials, while SMT solvers encounter NP-hard constraint resolution challenges when verifying dense representations, such as Neural Lyapunov Functions (NLFs) \cite{Formal_Synthesis_of_Lyapunov_Neural_Networks,richards2018,Manek2019,Mukherjee2022,chang_roohi_gao_2019, Towards_Learning_and_Verifying_Maximal_Neural_Lyapunov_Functions, Dai2021LyapunovstableNC, Wang2024}. {At the same time grid-based or branch-and-bound methods \cite{pmlr-v235-yang24f, giesl_hafstein_2015, 9682880} suffer from the curse of dimensionality.}

To mitigate these scaling limitations, recent literature has focused on exploiting inherent system properties. Algebraic verification can be rendered more tractable by leveraging sparsity patterns or chordal graph structures \cite{zheng2018sparse}. Concurrently, machine learning approaches have introduced compositional Neural Lyapunov architectures for interconnected subsystems \cite{liu2024compositionally}, or architectures that enforce certification by bounding the Lipschitz constant during training \cite{gaby2022lyapunov}, which shifts the curse of dimensionality to the sampling phase. 

Despite these advances, soundly certifying high-dimensional dense nonlinear systems remains unfeasible. To address this limitation and bypass the dimensional scaling bottlenecks of deterministic verification, we propose SCORE (\textbf{S}tatistical \textbf{C}ertification \textbf{o}f \textbf{R}egions of Attraction via \textbf{E}xtreme Value Theory). This framework shifts from deterministic guarantees to bounding the worst-case safety violation with high statistical confidence by reframing ROA certification as a constrained extreme-value estimation problem.

Specifically, the main contributions of this paper are:
\paragraph{\textbf{A Novel Statistical Framework}}
We integrate EVT with PSGLD to reframe ROA verification as a constrained extreme-value estimation problem. Under stationary sampling and a regular local-geometry condition around the global maximum, we prove that the Lyapunov derivative belongs to the Weibull maximum domain of attraction, yielding a finite-endpoint statistical estimator.
\paragraph{\textbf{Improved Scalability}} We empirically validate that this approach closely approximates the certification tightness of exact SOS programming while also scaling to dense systems of up to 500 dimensions, far exceeding typical formal verification limits.

The remainder of this paper is organized as follows: Section \ref{sec:problem_formulation} formulates the certification problem. Section \ref{sec:methodology} introduces the EVT certification framework, including the Langevin sampling strategy and theoretical proofs regarding the Weibull domain assumption. Section \ref{sec:experiments} presents numerical experiments benchmarking the tightness and scalability of the proposed method against state-of-the-art baselines. Finally, Section \ref{sec:conclusion} concludes the paper.

\section{PROBLEM FORMULATION}
\label{sec:problem_formulation}

Consider a continuous-time system of ODEs $\frac{dx}{dt}=f(x)$, where $x \in \mathbb{R}^{N}$ is the state vector and $f$ is locally Lipschitz continuous with an asymptotically stable origin ($f(0)=0$). To estimate its Region of Attraction (ROA), let $V:\mathcal{D} \to \mathbb{R}_{\geq 0}$ be a continuously differentiable Lyapunov candidate function defined on a domain $\mathcal{D} \subset \mathbb{R}^N$ where the origin is an asymptotically stable equilibrium, such that $V(0) = 0$ and $V(x) > 0$ for all $x \in \mathcal{D} \setminus \{0\}$. We choose a $\rho > 0$ such that the connected component of the sublevel set $\Omega_{\rho} = \{x \in \mathcal{D} \mid V(x) \le \rho\}$ containing the origin is closed and bounded, forming a compact volume.


By standard Lyapunov theory~\cite{khalil2002nonlinear}, $\Omega_{\rho}$ is guaranteed to be forward-invariant if $\dot{V}(x) \le 0$ for all $x \in \Omega_{\rho}$. To establish the strict decrease condition without the singular behavior at the origin, we assume that a smaller sublevel set $\Omega_{\rho_0}$, with $0<\rho_0<\rho$, has been independently certified by local analysis such that $\dot{V}(x)<0$ for all $x\in\Omega_{\rho_0}\setminus\{0\}$. We then define the compact verification domain
\[
\mathcal{A}_{\rho,\rho_0}:=\Omega_\rho\setminus\operatorname{int}(\Omega_{\rho_0}).
\]
Since $\dot{V}$ is continuous and $\mathcal{A}_{\rho,\rho_0}$ is compact, the worst-case violation
\[
\gamma^\star=\max_{x\in\mathcal{A}_{\rho,\rho_0}}\dot{V}(x)
\]
is well defined. If $\gamma^\star<0$, then together with the local certificate on $\Omega_{\rho_0}$, the strict decrease condition holds on all of $\Omega_\rho\setminus\{0\}$ and $\Omega_\rho$ is formally certified. Finding the exact global maximum $\gamma^\star$ is computationally intractable in high dimensions due to the highly non-convex nature of $\dot{V}(x)$. To bypass the current scaling limits, we instead estimate an upper bound on $\gamma^\star$ statistically.

\section{STATISTICAL CERTIFICATION VIA EXTREME VALUE THEORY}
\label{sec:methodology}

This section outlines our statistical framework for certifying the ROA for high-dimensional systems. By combining constrained Langevin sampling with Extreme Value Theory (EVT), we estimate the maximum $\gamma^\star$ on the compact verification domain $\mathcal{A}_{\rho,\rho_0}$. By analyzing the geometric properties of the optimization landscape near local maxima, we provide a theoretical justification for modeling the tail distribution of $\dot{V}$ as a Weibull-class Generalized Extreme Value (GEV) distribution.

While the Central Limit Theorem characterizes sample averages, Extreme Value Theory governs the asymptotic behavior of sample extremes. By the Fisher-Tippett-Gnedenko theorem \cite{haan_ferreira_2020}, the maximum of independent identically distributed (i.i.d.) random variables converges to the GEV distribution. The GEV tail behavior is dictated by its shape parameter $\xi$, partitioning it into the Fréchet ($\xi > 0$), Gumbel ($\xi = 0$), and Weibull ($\xi < 0$) maximum domains of attraction. Crucially, the Weibull domain is the only GEV family characterized by a \emph{strictly finite upper endpoint}. Because $\dot{V}(x)$ is continuous on the compact verification domain $\mathcal{A}_{\rho,\rho_0}$, its distribution has a finite right endpoint. Under the local polynomial-growth condition introduced below, the probability mass near this endpoint satisfies the regular-variation condition required for the Weibull maximum domain of attraction. The corresponding finite upper support, defined by the GEV parameters as $z^* = \mu - \sigma/\xi$, therefore provides a statistical estimator of the worst-case safety violation under the stated sampling assumptions.

\subsection{The Optimizer as a Dynamical System}

To sample the local geometry of the peaks of $\dot{V}$, we employ PSGLD \cite{pmlr-v134-lamperski21a, welling2011bayesian}. This models the optimization as a discrete-time stochastic process restricted to the defined volume:
\begin{equation}
    x_{k+1} = \Pi_{\mathcal{A}_{\rho,\rho_0}} \left( x_k + \eta \nabla \dot{V}(x_k) + \sqrt{2T \eta} \xi_k \right)
    \label{eq:disc_time_sgld}
\end{equation}
where $\Pi_{\mathcal{A}_{\rho,\rho_0}}$ denotes the projection operator onto the compact verification domain, $\eta$ is the step size, $T$ is the temperature parameter, and $\xi_k \sim \mathcal{N}(0, I_N)$ is injected standard Gaussian noise. 
In the idealized continuous-time stationary limit, the corresponding Langevin process admits the Gibbs measure used in the EVT analysis below. The finite-time algorithm therefore inherits the theoretical guarantee only when its retained samples are sufficiently close to this stationary distribution.

\subsection{Theoretical Justification for the Weibull Domain}
\label{sec:proof}
The central requirement for applying EVT is determining the correct maximum domain of attraction for the distribution of $\dot{V}$ samples. Under stationary Gibbs sampling and the regular local-geometry condition stated below, we show that the distribution of $\dot{V}$ belongs to the Weibull maximum domain of attraction.

Let $\dot{V}: \mathcal{A}_{\rho,\rho_0}\to\mathbb{R}$ represent the Lyapunov derivative on the compact verification domain. By continuity, $\dot{V}$ attains a finite global maximum $\gamma^\star=\dot{V}(x^\star)$ at some $x^\star\in\mathcal{A}_{\rho,\rho_0}$.

\begin{assumption}[Stationary Gibbs Measure]
\label{ass:1}
Assume $X$ is sampled from the stationary Gibbs measure induced by the idealized Langevin dynamics on $\mathcal{A}_{\rho,\rho_0}$.
\end{assumption}

Assumption~\ref{ass:1} describes the stationary limit. For finite-time PSGLD, a rigorous probabilistic guarantee additionally requires a quantitative mixing bound ensuring that the retained samples are sufficiently close to this measure. Without such a bound, the SCORE decision should be interpreted as a conditional statistical certificate, contingent on adequate mixing.

More generally, if the joint law of the retained samples is within total-variation distance $\delta$ of the stationary law, any probability statement established under the stationary model is weakened by at most $\delta$.

\begin{assumption}[Generalized Local Geometry]
\label{ass:quadratic_geometry}
Let $x^\star\in\mathcal{A}_{\rho,\rho_0}$ be the unique global maximizer of $\dot{V}$. In local feasible coordinates around $x^\star$, assume
\[
\dot{V}(x^\star)-\dot{V}(x)
=
h(x-x^\star)+o(\|x-x^\star\|^p),
\]
where $p>0$ and $h$ is continuous, positive away from the origin, and positively homogeneous of degree $p$, i.e., $h(a u)=a^p h(u)$ for $a>0$. We additionally assume that the feasible set is locally regular at $x^\star$, so that its rescaled local geometry converges to a tangent cone $C=T_{\mathcal{A}_{\rho,\rho_0}}(x^\star)$ having positive $N$-dimensional solid angle.
\end{assumption}


\begin{remark}[Scope of the Local Geometry Assumption]
Assumption 2 includes non-degenerate quadratic maxima ($p=2$) and isolated finite-order degeneracies, but excludes genuinely flat directions, plateaus, and non-isolated maximizers. The unique-maximizer formulation is used for simplicity; the same argument extends to finitely many isolated global maximizers sharing the same local order.
\end{remark}

\begin{lemma}[Polynomial Scaling of the Optimality Gap]
\label{lem:gap_scaling}
Let $Z:=\dot{V}(x^\star)-\dot{V}(X)$ denote the optimality gap.
Under Assumptions~\ref{ass:1} and~\ref{ass:quadratic_geometry}, as
$z\to0^+$,
\[
    \mathbb{P}(Z\le z)\sim K z^{N/p}
\]
for some constant $K>0$.
\end{lemma}

\begin{proof}
Let
\[
B_z:=\{x\in\mathcal{A}_{\rho,\rho_0}:
\dot{V}(x^\star)-\dot{V}(x)\le z\}.
\]
Let $C=T_{\mathcal{A}_{\rho,\rho_0}}(x^\star)$ denote the local
tangent cone and let
$S_C=C\cap\mathbb{S}^{N-1}$ be its feasible angular section.
For $x=x^\star+r\theta$, $\theta\in S_C$, Assumption~\ref{ass:quadratic_geometry} gives
\[
\dot{V}(x^\star)-\dot{V}(x)
=r^p\big(h(\theta)+o(1)\big),
\qquad r\to0^+,
\]
where we used the $p$-homogeneity of $h$. Since $h$ is continuous
and strictly positive on the compact set $S_C$, there exist
$0<m\le M<\infty$ such that
$m\le h(\theta)\le M$ on $S_C$.

Consequently, for every sufficiently small $\varepsilon\in(0,m)$
and all sufficiently small $r$,
\[
r^p(h(\theta)-\varepsilon)
\le
\dot{V}(x^\star)-\dot{V}(x)
\le
r^p(h(\theta)+\varepsilon).
\]
Thus, up to an $o(z^{N/p})$ contribution from the local boundary
geometry, polar integration yields
\[
\int_{S_C}(h(\theta)+\varepsilon)^{-N/p}\,d\theta
\le
\frac{N\lambda(B_z)}{z^{N/p}}
\le
\int_{S_C}(h(\theta)-\varepsilon)^{-N/p}\,d\theta .
\]
Letting $\varepsilon\to0$ gives
\[
\lambda(B_z)\sim C_h z^{N/p},
\qquad
C_h:=\frac{1}{N}\int_{S_C}h(\theta)^{-N/p}\,d\theta\in(0,\infty).
\]

Under Assumption~\ref{ass:1}, the stationary Gibbs density is
$q(x)\propto\exp(\dot{V}(x)/T)$ on
$\mathcal{A}_{\rho,\rho_0}$. It is continuous and strictly positive
at $x^\star$. Since $B_z$ shrinks to $x^\star$ as $z\to0^+$,
$q(x)=q(x^\star)+o(1)$ uniformly on $B_z$. Therefore,
\[
\mathbb{P}(Z\le z)
=\int_{B_z}q(x)\,dx
\sim q(x^\star)C_h z^{N/p}
=:Kz^{N/p},
\]
where $K>0$.
\end{proof}

\begin{theorem}[Weibull Maximum Domain of Attraction]
\label{thm:weibull}
Let $Y = \dot{V}(X)$ be the random variable representing the sampled evaluations of the Lyapunov derivative. Under Assumptions \ref{ass:1} and \ref{ass:quadratic_geometry}, the distribution of $Y$ strictly belongs to the Weibull Maximum Domain of Attraction.
\end{theorem}



\begin{proof}
By Lemma~\ref{lem:gap_scaling}, for every $x>0$,
\[
\lim_{t\to0^+}
\frac{1-F_Y(\gamma^\star-tx)}
     {1-F_Y(\gamma^\star-t)}
=
\lim_{t\to0^+}
\frac{K(tx)^{N/p}}{Kt^{N/p}}
=x^{N/p}.
\]
This is the endpoint regular-variation criterion for the Weibull
maximum domain of attraction. Hence $-1/\xi=N/p$, or
$\xi=-p/N<0$.
\end{proof}



\begin{corollary}[Finite Endpoint and Statistical Certification]
\label{cor:certification}
Under Theorem~\ref{thm:weibull}, normalized block maxima of $\dot{V}(X)$ converge, as the block size increases, to a Weibull-class GEV distribution with finite endpoint $\gamma^\star$. Hence, if a statistical upper confidence bound $U$ for $\gamma^\star$ satisfies $U<0$, then together with the local certificate on $\Omega_{\rho_0}$, $\Omega_\rho$ is statistically certified at the stated confidence level, conditional on the sampling and inference assumptions.
\end{corollary}

This result provides the theoretical justification for enforcing a negative GEV distribution shape parameter ($\xi < 0$), formally capturing the physical bounding of the system's energy dissipation on the constraint volume.

\subsection{The Statistical Certification Algorithm}
With the theoretical framework established, the empirical certification is conducted using a block maxima approach:

\paragraph{\textbf{Sampling}} Simulate $M$ independent PSGLD chains uniformly initialized across $\mathcal{A}_{\rho,\rho_0}$ to collect $\dot{V}$ samples.
 
\paragraph{\textbf{Block Maxima}} Divide the optimization results into $N_{blocks}$ non-overlapping blocks of size $m$ (such that $M = N_{blocks} \times m$). Extract the single maximum value from each block, $M_b = \max \{y^{(1)}_{final}, \dots, y^{(m)}_{final}\}$, discarding the rest. This constructs the extreme-value dataset, denoted as $S_{\max}$, consisting of $N_{blocks}$ i.i.d. maxima ready for statistical fitting.
\paragraph{\textbf{GEV Fitting}} By the Fisher-Tippett-Gnedenko theorem \cite{coles_2001}, $S_{\max}$ is fitted to the GEV distribution using Maximum Likelihood Estimation (MLE) to estimate the defining shape, location, and scale parameters, denoted by the tuple $(\hat{\xi}, \hat{\mu}, \hat{\sigma})$.
Because this MLE step is a low-dimensional (3-parameter) optimization over a 1D scalar dataset, it is solved efficiently using standard solvers without introducing high-dimensional approximation errors.
\paragraph{\textbf{Confidence Estimation}}
\label{ci}
Because the distribution belongs to the Weibull domain ($\hat{\xi} < 0$), it possesses a finite right endpoint. Using the fitted parameters, this theoretical maximum of the Lyapunov derivative is evaluated as $z^* = \hat{\mu} - \hat{\sigma}/\hat{\xi}$. To account for finite-sample uncertainty in this point estimate, we construct a bootstrap upper confidence bound, denoted $CI_{upper}$, for $z^*$ via a non-parametric empirical bootstrap. We generate $N_{\text{boot}}$ synthetic datasets by resampling the empirical block maxima $S_{\text{max}}$ with replacement. For each resample, we compute the MLE fit and its corresponding finite upper bound. By extracting the $(1-\alpha)$ percentile of these bootstrapped bounds, we obtain an upper confidence bound for the safety violation under the fitted model. Finally, before the sublevel set is accepted by SCORE, we enforce a two-part condition: the worst-case statistical bound must be strictly negative ($CI_{upper} < 0$), and the empirical dataset $S_{\max}$ must also pass a Kolmogorov-Smirnov (KS) goodness-of-fit test as a diagnostic for empirical compatibility with the fitted GEV model.

\subsection{Maximal ROA via Binary Search}
Algorithm \ref{alg:certification} evaluates a single, fixed sublevel set $\Omega_\rho$. However, our ultimate objective is to maximize the certified volume of the ROA. Because the sublevel sets are strictly nested (for $\rho_1 < \rho_2$, $\Omega_{\rho_1} \subset \Omega_{\rho_2}$), we can efficiently approximate the maximal certified level set by wrapping the statistical certification engine in a binary search. 

Taking the independently certified inner level $\rho_0$ as the initial safe lower bound $\rho_{low}$ and an upper bound $\rho_{high}$, we iteratively evaluate the midpoint $\rho_{mid}$. If the EVT framework successfully certifies $\Omega_{\rho_{mid}}$ with the required statistical confidence, we update $\rho_{low} = \rho_{mid}$. If the certification is rejected, either due to a concrete counterexample or a positive upper confidence bound, we update $\rho_{high} = \rho_{mid}$.

Furthermore, this bisection strategy benefits from an inherent computational asymmetry: while successfully certifying a conservative $\Omega_{\rho_{mid}}$ requires the full PSGLD sampling budget to estimate the EVT bound, rejecting an overly aggressive level set often terminates almost instantaneously upon the first empirical discovery of a positive Lie derivative, drastically accelerating the overall search.

\begin{algorithm}
\caption{Stochastic EVT Certification via PSGLD}
\label{alg:certification}
\begin{algorithmic}[1]
\renewcommand{\algorithmicrequire}{\textbf{Input:}}
\renewcommand{\algorithmicensure}{\textbf{Output:}}
\REQUIRE ODE field $f$, Lyapunov $V$, Levels $\rho_0<\rho$, Block size $m$, Noise $T$, Step $\eta$, Confidence $1-\alpha$, $N_{blocks}$
\ENSURE SCORE Decision
\STATE $S_{\max} \gets \emptyset$
\FOR{$b = 1$ to $N_{blocks}$}
    \STATE Sample batch $\{x^{(1)}_0, \dots, x^{(m)}_0\}$ uniformly on $\mathcal{A}_{\rho,\rho_0}$.
    \FOR{$i = 1$ to $m$}
        \STATE $x \gets x^{(i)}_0$
        \FOR{$k = 1$ to $K$}
            \STATE Sample $\xi \sim \mathcal{N}(0, I_N)$
            \STATE $\tilde{x} \gets x + \eta \nabla \dot{V}(x) + \sqrt{2T\eta} \xi$
            \STATE $x \gets \Pi_{\mathcal{A}_{\rho,\rho_0}}(\tilde{x})$
        \ENDFOR
        \STATE $y^{(i)}_{final} \gets \dot{V}(x)$
    \ENDFOR
    \STATE $M_b = \max \{ y^{(1)}_{final}, \dots, y^{(m)}_{final} \}$
    \STATE $S_{\max} \gets S_{\max} \cup \{ M_b \}$
\ENDFOR
\STATE $(\hat{\xi}, \hat{\mu}, \hat{\sigma}) \gets \text{MLE}(S_{\max})$
\IF{$\hat{\xi} \geq 0$} 
    \RETURN \textbf{FAIL} (Heavy tail detected) 
\ENDIF
\STATE $z^* \gets \hat{\mu} - \frac{\hat{\sigma}}{\hat{\xi}}$ \COMMENT{Upper bound of support}
\STATE $CI_{upper} \gets \text{Bootstrap}(z^*, 1-\alpha)$
\IF{$CI_{upper} < 0$ \AND $\text{GoodnessOfFit}(S_{\max})$ is True} 
    \RETURN \textbf{ACCEPTED} 
\ELSE
    \RETURN \textbf{REJECTED}
\ENDIF
\end{algorithmic}
\end{algorithm}

\subsection{Assumption Analysis and Failure Modes}
The validity of this statistical certification hinges on several practical considerations:

\paragraph{\textbf{Sufficient Coverage and Mixing}}
The PSGLD chains must sufficiently explore the regions contributing to the upper tail of $\dot{V}$. This requirement is exact under Assumption 1 but cannot in general be certified from finite trajectories without a quantitative mixing-time analysis. Uniform initialization and reseeding improve empirical coverage but do not provide a formal guarantee. Likewise, the KS test is only a goodness-of-fit diagnostic: rejection indicates inconsistency with the fitted model, whereas acceptance does not prove stationarity or adequate state-space coverage. Consequently, absent a problem-specific mixing bound, SCORE provides a conditional statistical certificate rather than an unconditional probabilistic guarantee.

\paragraph{\textbf{Asymptotic Convergence}} The block size $m$ must be large enough to converge to the GEV limit.

\paragraph{\textbf{Compactness}} We explicitly restrict attention to levels $\rho$ for which the connected component $\Omega_\rho$ containing the origin is closed and bounded. Radial unboundedness of $V$ is sufficient, but not necessary, for this property. 

\begin{remark}[Practical Sampling Heuristics]
Practical heuristics such as reseeding or gradient clipping may perturb the ideal stationary sampling distribution. Their use therefore reinforces the conditional nature of the finite-time statistical guarantee unless their effect is explicitly controlled.
\end{remark}

\section{NUMERICAL EXPERIMENTS}
\label{sec:experiments}

To evaluate the proposed statistical certification framework, we conduct two distinct numerical experiments. First, we assess the \textit{tightness} of the statistical certification against exact deterministic verification methods on a standard 2D nonlinear benchmark. Second, we evaluate the computational \textit{scalability} of the proposed framework against state-of-the-art baselines on a dense dissipative system of increasing state-space dimensionality, up to $N=500$. {Experiments were run in Python on an Intel Core Ultra 9 and NVIDIA RTX 5090 Source code is available at 
\footnote{\url{https://github.com/SOLARIS-JHU/SCORE}}.}

\subsection{Parameterization of the Lyapunov Candidate}
While the primary focus of this work is the statistical \textit{verification} engine rather than the \textit{synthesis} of Lyapunov functions, evaluating the framework necessitates a baseline candidate function. For computational efficiency, we avoid deep neural networks and parameterize the Lyapunov candidate using a dictionary-based quadratic form (Gram matrix formulation): $V(x) = z(x)^\top Q z(x)$, where $z(x)$ is a generic library vector of smooth, linearly independent basis functions, and $Q \succ 0$ is a symmetric positive-definite parameter matrix. The matrix $Q$ is optimized offline via gradient descent, minimizing a loss function that applies a penalty only when the strict decrease condition is violated and $\dot{V} + \epsilon V > 0$. 

We emphasize that synthesizing optimal Lyapunov candidates in high-dimensional spaces remains an open research challenge; this specific formulation is employed strictly for its simplicity, computational efficiency.

\subsection{Tightness of Certification: 2D Van der Pol Oscillator}
We first consider the 2D reversed Van der Pol oscillator \cite{khodadadi_samadi_khaloozadeh_2014}, a well-known non-convex benchmark system featuring a stable limit cycle. {While finding the maximal ROA analytically is generally impossible, for this specific 2D system, the exact boundary of the maximal ROA is defined by its limit cycle.} The objective is to compare the state-space volume of the certified ROA, denoted $\Omega_\rho$, achieved by various verification techniques. 

We benchmark the proposed EVT approach (evaluated at a 99\% confidence level) against classical SOS programming, as well as formal verification techniques utilizing SMT applied to Neural Lyapunov Functions (NLF) \cite{chang_roohi_gao_2019}, Input Convex Neural Networks (ICNN) \cite{amos2017input}, and Physics-Informed Neural Networks (PINN) for Zubov equation~\cite{LIU2025112193}. 

As illustrated in Fig. \ref{fig:vdp_level_sets}, the standard SOS method yields the tightest ROA estimate, certifying $95\%$ of the state-space area enclosed by the limit cycle. {A common pitfall in evaluating verification methods is conflating the quality of the synthesized Lyapunov candidate with the tightness of the certification engine itself.} To isolate and evaluate the certification capability of our statistical method independent of the candidate function's quality, we applied the EVT certification directly to the SOS-derived polynomial candidate (denoted EVT + SOS). This yielded a certified ROA of $83\%$, demonstrating that the EVT {statistical} bound closely approximat{es} exact deterministic certification {when the candidate parameterization is held constant}. Conversely, when applying EVT to a candidate synthesized via our dictionary-based approach (denoted EVT + Dict-Gram), the certified {volume} is notably more conservative ($37\%$). This dichotomy highlights a critical insight: while the EVT framework serves as a robust and scalable \textit{certification} engine, the synthesis of optimal candidate functions in non-convex landscapes remains a distinct and open research direction.

\begin{figure}
    \centering
    \includegraphics[width=1\linewidth]{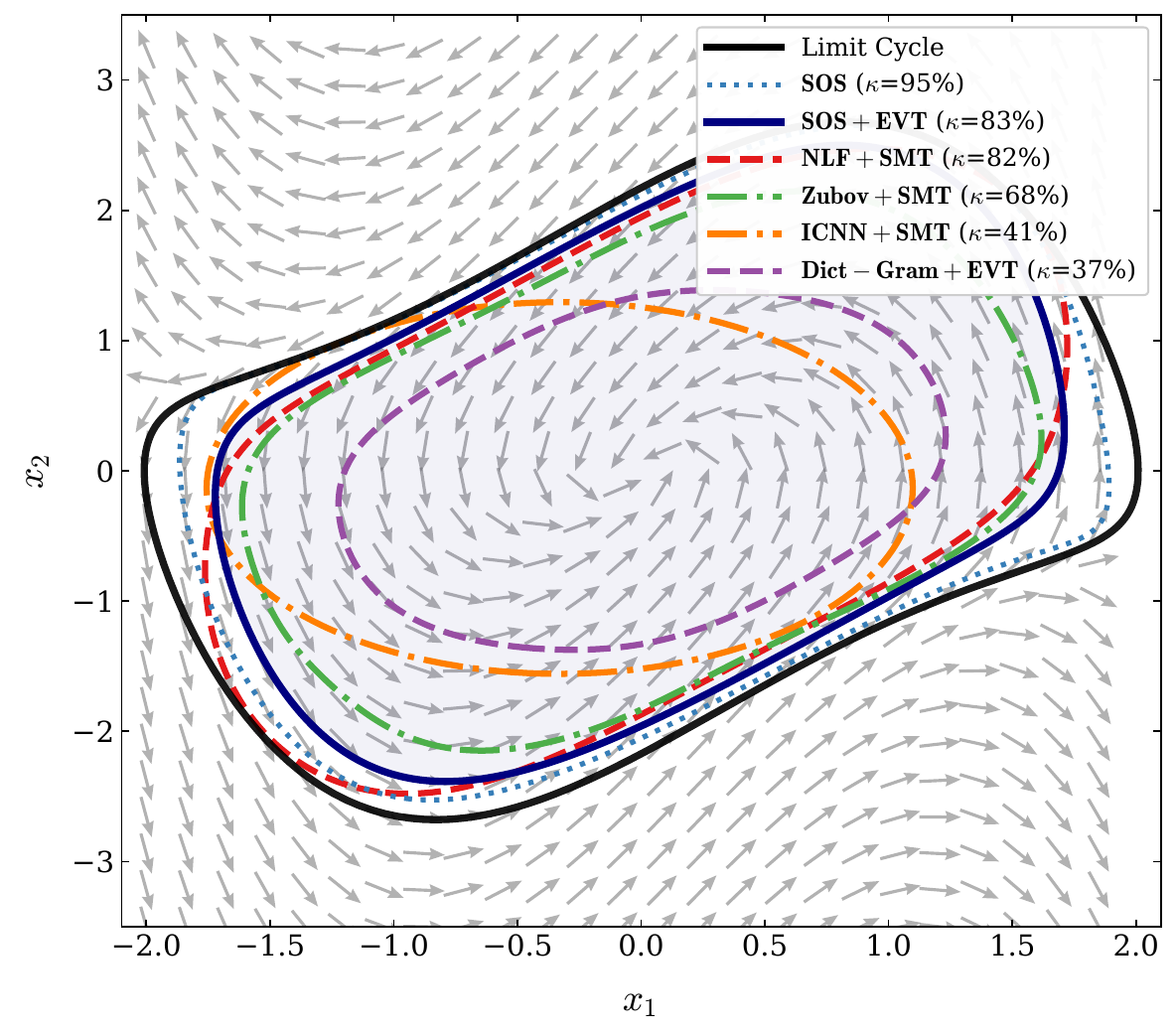}
    \caption{Certified Regions of Attraction for the 2D Van der Pol oscillator. Because our primary contribution is the scalable \textit{certification} engine rather than Lyapunov \textit{synthesis}, we evaluate EVT on two distinct candidates. While our proposed pipeline (EVT + Dict-Gram) yields a conservative ROA due to the simple dictionary-based synthesis, applying our EVT certification to an SOS-synthesized candidate (solid dark blue) achieves an ROA comparable to exact SOS formal verification (dotted blue). This isolates the conservatism strictly to the synthesis step and demonstrates the tightness of our statistical bound. The parameter $\kappa$ denotes the percentage of the true ROA successfully certified.}
    \label{fig:vdp_level_sets}
\end{figure}

\subsection{Scalability to High Dimensions}
While exact formal methods provide tight guarantees in low dimensions, their primary limitation is the curse of dimensionality. To evaluate scalability, we designed a dense, linear dissipative system governed by $\dot{x} = M x$, where $M \in \mathbb{R}^{N \times N}$ is a densely populated Hurwitz matrix constructed via orthogonal transformations. 
For the proposed approach, the smooth analytical nature of the dictionary-based parameterization introduced in Section IV-A is critical for maintaining computational tractability. By avoiding the vanishing gradients and backpropagation overhead inherent to deep neural networks, the computation of the Lie derivative and its gradients during PSGLD exploration remains highly stable and computationally efficient, even for a state dimension of $N=500$. A strict computational timeout of 10 minutes per certification task was enforced across all evaluated methods.

\begin{table}[ht]
\caption{Maximum Scalability within 10-Minute Timeout}
\label{tab:scalability}
\centering
\begin{tabular}{lcc}
\hline
\textbf{Certification Pipeline} & \begin{tabular}[c]{@{}c@{}}\textbf{Max Accepted} \\ \textbf{Dimension ($N$)}\end{tabular} \\
\hline
Zubov (PINN) + SMT              & 2 \\
ICNN + SMT                      & 2 \\
NLF + SMT                       & 10 \\
Classical SOS                   & 20 \\
\textbf{Dict-Gram + EVT (Proposed)} & \textbf{500} \\
\hline
\end{tabular}
\end{table}

The results, summarized in Table \ref{tab:scalability}, highlight the severe computational limitations of existing deterministic verification pipelines. Techniques coupling neural network architectures (ICNN and PINN for Zubov PDE) with SMT solvers fail to scale beyond $N=2$. 
as dense parameters and nonlinear activations cause SMT solvers to time out due to combinatorial complexity. Specifically, ICNNs create intractable constraints through skip-connections and Softplus activations, while PINNs (for the Zubov equation) force solvers to balance rigid PDE constraints with Sigmoid outputs. NLF formulations scale marginally better ($N=10$) by using a shallow structure and simple algebraic positivity ($V = \text{out}^2 + \epsilon \|x\|^2$). Meanwhile, algebraic SOS programming succumbs to matrix scaling limits at $N=20$.
In stark contrast, the proposed Dict-Gram + EVT approach successfully completed the SCORE verification procedure for the $N=500$ dimensional system. Because the Langevin-based stochastic exploration scales gracefully with the ambient state-space dimension, relying exclusively on local gradient evaluations rather than global algebraic or logical constraint resolution, it effectively bypasses the fundamental combinatorial bottlenecks that restrict traditional formal verification.

\section{Conclusion}
\label{sec:conclusion}
This paper introduced a novel statistical framework for certifying the ROA of high-dimensional nonlinear systems, bypassing the severe scaling limitations of classical deterministic verification. By framing the certification as a rare-event simulation, we demonstrated that the distribution of extreme values explored by PSGLD on a constraint manifold theoretically belongs to the Weibull domain of attraction. Under the stationary-sampling assumption, this permits statistical upper bounding of the Lyapunov derivative. 
While exact formal methods like SOS and SMT-based solvers succumb to combinatorial explosions and architectural complexities between 2 and 20 dimensions, the EVT-based approach successfully applied the proposed statistical verification procedure to a dense 500-dimensional linear system well within a 10-minute computational budget. 
Future work will focus on two primary directions.
First, as highlighted by the 2D benchmark, developing more expressive and optimal construction methods for Lyapunov candidates remains an open challenge due to the ill-posedness of the problem and the complex optimization landscape. Second, we aim to extend this framework to infinite-dimensional systems, such as discretized PDEs. 
Such systems introduce extreme ill-conditioning into the optimization landscape, necessitating preconditioning and more advanced Monte Carlo Markov Chain algorithms to maintain efficient mixing in high dimensions.

\bibliographystyle{IEEEtran}
\bibliography{bib}


\end{document}